\documentclass[11pt,journal,twocolumn]{IEEEtran}

\usepackage{amsmath,amssymb,amsfonts,bm}
\usepackage{amsthm}

\usepackage{graphicx}
\usepackage{float}
\usepackage{subfigure}
\usepackage{empheq}
\usepackage{algorithm}
\usepackage{algorithmic}

\usepackage{cite}
\usepackage{color}
\usepackage{times}

\theoremstyle{plain}
\newtheorem{theorem}{Theorem}[section]

\newtheorem{corollary}[theorem]{Corollary}

\theoremstyle{definition}
\newtheorem{definition}[theorem]{Definition}

\theoremstyle{remark}
\newtheorem{remark}[theorem]{Remark}

\theoremstyle{plain}

\begin{document}

\title{Complementary Waveform Coordination for Doppler-Resilient MIMO Radar}

\author{Maksim Antonik,  Wenbing Dang, Ali Pezeshki, Stephen D. Howard, and William Moran
\thanks{This work is supported in part by the AFOSR under award FA9550-24-1-0326. M. Antonik and A. Pezeshki are with Colorado State University. W. Dang is with Nvidia. S. D. Howard is with DSTG, Australia. W. Moran is with Melbourne University, Australia.}}
\maketitle

\begin{abstract}
Complementary waveform libraries yield impulse-like aggregate delay
responses at zero Doppler, but pulse-dependent Doppler phases degrade
their sidelobe cancellation. This paper develops a linear-algebraic
formulation for coordinating a (D)-ary transmission schedule and
slow-time receive weights for paraunitary MIMO waveform libraries.
Using a modal decomposition of the CPI ambiguity matrix,
prescribed-order Doppler-nulling conditions for the (D-1) nonzero
modes are expressed as homogeneous linear constraints on the receive
weights for a fixed schedule. This representation gives the
SNR-optimal fixed-schedule weights by subspace projection and reduces
the remaining design to a discrete schedule search. Degrees-of-freedom
and modal-growth analyses relate the design to the CPI length, null
order, and waveform dimension. Numerical examples compare the
resulting designs with established binary constructions and demonstrate
Doppler suppression across the full ambiguity matrix of a four-waveform
configuration. 
\end{abstract}


\section{Introduction}

Phase coding is widely used in radar pulse compression to achieve fine range resolution while maintaining sufficient transmitted energy \cite{Levanon-book}. A prominent class of phase-coded waveforms is based on Golay complementary sequences \cite{Golay-IRE61}. The aperiodic autocorrelation functions of the two sequences in a Golay pair sum to zero at every nonzero delay. Consequently, when waveforms coded by the pair are transmitted over separate pulse repetition intervals (PRIs) and their matched-filter outputs are coherently combined, the aggregate response is impulse-like in delay along the zero-Doppler axis. This complementary property was subsequently generalized from sequence pairs to larger complementary sets \cite{Tseng-IT72}.

Colocated MIMO radar uses multiple transmit channels and waveform diversity to provide additional spatial degrees of freedom \cite{Stoica-SPMSept07}. This multichannel setting naturally motivates waveform families with both complementary autocorrelation and cross-channel orthogonality. Complete complementary codes (CCCs) were originally developed for synchronous multiuser spread-spectrum communications \cite{Suehiro-IT88}. They extend complementarity to families of sequence sets whose summed autocorrelation sidelobes and inter-set cross-correlations vanish \cite{Han-IT11}. Specifically, the summed aperiodic autocorrelation sidelobes vanish within each complementary set, while the summed cross-correlations between distinct sets vanish at every delay. These properties make CCCs a natural foundation for MIMO radar waveform design, where an impulse-like aggregate delay response must be maintained while separating multiple transmit channels.

The cancellation provided by complementary waveforms is highly sensitive to Doppler because the pulse-dependent phase progression disrupts the coherent sums required by the complementary identities. Early Doppler-resilient designs retained matched receive processing and modified only the slow-time ordering of the transmitted waveforms. The Prouhet--Thue--Morse (PTM) sequence \cite{Allouche-SETA98} was used to schedule Golay-coded pulses, imposing a high-order spectral null at zero Doppler \cite{Pezeshki-IT08}. Related transmit-only constructions employed first-order Reed--Muller codes, oversampled PTM sequences, and generalized Thue--Morse or equal-sums-of-powers sequences to control the locations and orders of Doppler nulls \cite{Pezeshki-Asil07,Pezeshki-WDD09,Nguyen-IET16}. Generalized PTM scheduling was subsequently applied to CCCs in MIMO radar \cite{Tang_PTM}.

Another line of work jointly designed the transmission schedule and receive weights. In our prior work, we introduced the \((P,Q)\) framework, which coordinates these quantities to create high-order Doppler nulls and suppress range sidelobes over an interval about zero Doppler \cite{Dang_Asilomar11,Dang_arXiv20}. Subsequent work reformulated the binary Golay problem using semidefinite optimization, null-space methods, and SVD-based methods for sidelobe suppression over prescribed Doppler intervals \cite{Wu-SPL20,Wang-arXiv21,Wang-IET23}.

Building on the general \(D\)-ary modal decomposition developed in our prior work \cite{Dang_arXiv20}, this paper systematizes the joint schedule-and-weight design for complementary waveform sets and paraunitary MIMO waveform matrices. For an arbitrary (D)-ary schedule, the moment-nulling conditions for all nonzero modes are assembled into a single linear system. This formulation provides the optimal receive weights for a fixed schedule and reduces the joint design to a discrete optimization over the schedule. We further characterize the available slow-time degrees of freedom and derive a bound relating Doppler-clearance width to the null order, CPI length, and waveform dimension.

\section{Signal Model}
\label{sec:signal_model}
We consider a colocated MIMO radar with \(D\) transmit channels. During each PRI, the \(D\) transmitters simultaneously radiate the components of a \(D\)-channel waveform vector selected from a common library. The receiver applies the corresponding \(D\) matched filters, producing a \(D\times D\) matrix of auto- and cross-responses.

\subsection{Paraunitary Waveform Library}

We represent the transmit code library as a causal $D\times D$ matrix FIR system
\begin{equation}
\mathbf{S}(z) \triangleq \sum_{\ell=0}^{L-1}\mathbf{S}[\ell]\,z^{-\ell},
\qquad \mathbf{S}[\ell]\in\mathbb{C}^{D\times D}.
\label{eq:Sz_def}
\end{equation}
The system $\mathbf{S}(z)$ is \emph{paraunitary} if
$\mathbf{S}(z)\mathbf{S}^H(z^{-1}) = c\,\mathbf{I}_D$ for some constant
$c>0$. Expanding this condition in powers of $z$ yields the time-domain
correlation identity
\begin{equation}
\sum_{\ell=0}^{L-1}\mathbf{S}[\ell]\mathbf{S}^{H}[\ell-k]
=
c\,\mathbf{I}_D\,\delta[k],
\qquad k=-(L-1),\ldots,(L-1),
\label{eq:PU_time}
\end{equation}
where $\delta[k]$ is the Kronecker delta. Paraunitary matrix polynomial
systems arise naturally in multirate filter-bank theory
\cite{VaidyanathanBook}; several constructions of complementary waveform
libraries based on this framework have been proposed in the literature
(e.g., \cite{BudiSinSpasojevic}). The specific construction used in the
numerical studies of this paper is summarized in
Appendix~\ref{app:PUconstruction}; the analysis requires only that
$\mathbf{S}(z)$ satisfy \eqref{eq:PU_time}.

Let $\mathbf{s}_d[\ell]\in\mathbb{C}^D$ denote the $d$th column of
$\mathbf{S}[\ell]$, and define the per-column correlation matrices
\begin{equation}
\mathbf{C}_d[k]
\triangleq
\sum_{\ell=0}^{L-1}\mathbf{s}_d[\ell]\mathbf{s}_d^{H}[\ell-k].
\label{eq:Cd_def}
\end{equation}
Decomposing \eqref{eq:PU_time} by columns shows that
\begin{equation}
\sum_{d=0}^{D-1}\mathbf{C}_d[k]
=
c\,\mathbf{I}_D\,\delta[k].
\label{eq:comp_set}
\end{equation}
The FIR coefficients are taken to be unimodular,
$|S_{i,d}[\ell]|=1$ for all $i,d,\ell$, so that counting coefficients
gives a total energy of $D^2L$. Evaluating \eqref{eq:PU_time} at $k=0$
then yields $c=DL$.

\subsection{CPI Transmission and Ambiguity}

Each column $d$ of the library defines a $D$-channel transmit waveform vector
$\mathbf{x}_d(t)=[\,x_{0,d}(t),\ldots,x_{D-1,d}(t)\,]^T$, whose waveform on
transmit channel $i$ is
\begin{equation}
x_{i,d}(t)
=
\sum_{\ell=0}^{L-1}
S_{i,d}[\ell]\,\Omega(t-\ell T_c),
\qquad i,d=0,\ldots,D-1.
\label{eq:continuous_waveform}
\end{equation}
The chip waveform $\Omega(t)$ has duration $T_c$, unit energy,
$\int_{-T_c/2}^{T_c/2}|\Omega(t)|^2\,dt=1$, and mutually orthogonal shifts
at integer multiples of $T_c$.

Let $P=\{p_n\}_{n=0}^{N-1}$, $p_n\in\{0,\ldots,D-1\}$,
be a $D$-ary slow-time sequence that selects which column waveform is
transmitted during each PRI.  During the $n$th PRI the scheduling sequence selects column $d = p_n$, and
all $D$ channels simultaneously radiate
$\mathbf{x}_{p_n}(t) = \mathbf{x}_d(t)\big|_{d=p_n}$.

We consider a single point target with delay $\tau$ and Doppler frequency
$\nu$. The $n$th-PRI ambiguity matrix is
\begin{equation}
\boldsymbol{\chi}_n(\tau,\nu)
=
\int \mathbf{x}_{p_n}(t-nT)\,
\mathbf{x}_{p_n}^{H}(t-\tau-nT)\,e^{j\nu t}\,dt.
\label{eq:chi_n_cont}
\end{equation}
Define the intra-PRI fast time $t' = t - nT$, so that
\[
e^{j\nu t} = e^{j\nu nT}\,e^{j\nu t'}.
\]
Since $|t'| \leq LT_c$ and $|\nu|\,LT_c \ll 1$ (stop-and-hop
approximation~\cite{Levanon-book}), $e^{j\nu t'} \approx 1$ and thus
$e^{j\nu t} \approx e^{j\nu nT}$.

Substituting $e^{j\nu t} \approx e^{j\nu nT}$ into
\eqref{eq:chi_n_cont} gives
\begin{equation}
\boldsymbol{\chi}_n(\tau,\nu)
=
e^{j\nu nT}
\int \mathbf{x}_{p_n}(t')\,
\mathbf{x}_{p_n}^{H}(t'-\tau)\,dt'.
\label{eq:chi_n_stopandhop}
\end{equation}
Evaluating on the chip-delay grid $\tau = kT_c$, orthogonality of
integer chip shifts reduces the integral to
$\mathbf{C}_{p_n}[k]$. Defining $\theta \triangleq \nu T$,
\begin{equation}
\boldsymbol{\chi}_n(k,\theta)
=
e^{jn\theta}\,\mathbf{C}_{p_n}[k],
\label{eq:chi_n_discrete}
\end{equation}
where $\mathbf{C}_{p_n}[k]$ is drawn from
$\{\mathbf{C}_d[k]\}_{d=0}^{D-1}$ defined in \eqref{eq:Cd_def}.

\begin{remark}[Steering Vector]
For a fixed look direction, the transmit and receive steering terms reduce to constant matched-filter gains and do not affect the delay--Doppler structure.
\end{remark}

Let $Q=\{q_n\}_{n=0}^{N-1}$ denote a complex-valued slow-time weighting
sequence applied during CPI combining. 
Combining across the CPI with weighting sequence $Q$ yields the aggregate
ambiguity matrix
\begin{equation}
\boldsymbol{\chi}(k,\theta)
=
\sum_{n=0}^{N-1} q_n\,e^{jn\theta}\,\mathbf{C}_{p_n}[k].
\label{eq:chi_cpi_discrete}
\end{equation}

\subsection{Constant-Envelope Realization via MSK}
\label{sec:msk_realization}

The ideal square-chip realization in
\eqref{eq:continuous_waveform} has abrupt phase transitions and poor
spectral confinement. These effects are avoided using the
biphase-to-quadriphase (BTQ) transformation
\cite{TaylorBlinchikoff88}, which provides a constant-envelope,
phase-continuous MSK realization.

The BTQ transformation has been shown to preserve the complementary
correlation property of single-channel Golay pairs \cite{Dong2019};
Appendix~\ref{app:MSKrealization} shows that the same property holds
for the paraunitary waveform library used here. All subsequent
waveform-level simulations use this realization:
\begin{equation}
\begin{aligned}
x_{i,d}^{\mathrm{BTQ}}(t)
&=
\sqrt{T_c}\sum_{r=0}^{L/2-1}
\Bigg[
S_{i,d}[2r+1]
\cos\!\left(\frac{\pi t}{T_c}\right)
\Omega\bigl(t-(r+1)T_c\bigr)
\\
&\quad
-jS_{i,d}[2r]
\sin\!\left(\frac{\pi t}{T_c}\right)
\Omega\!\left(t-\left(r+\frac12\right)T_c\right)
\Bigg].
\end{aligned}
\label{eq:btq_waveform}
\end{equation}

\section{Doppler-Resilient Design in the Mode Domain}
\label{sec:mode_design}

\subsection{Mode-Domain Representation of the CPI Ambiguity}

Since $p_n\in\{0,\dots,D-1\}$, the per-PRI correlation matrix can be written as
\begin{equation}
\mathbf{C}_{p_n}[k]
=
\sum_{d=0}^{D-1}\delta_{d,p_n}\,\mathbf{C}_{d}[k].
\label{eq:Cpn_select}
\end{equation}

And because the column index $d$ takes values in the finite cyclic group
$\mathbb{Z}_D$, the Kronecker selector $\delta_{d,p_n}$ admits a Fourier
expansion in the orthogonal character basis $\{\omega^{rd}\}_{r=0}^{D-1}$.
Using the discrete Fourier transform (DFT) selector identity
\begin{equation}
\delta_{d,p_n}
=
\frac{1}{D}\sum_{r=0}^{D-1}\omega^{r(p_n-d)},
\qquad \omega=e^{j2\pi/D},
\label{eq:selector_identity}
\end{equation}
and reordering finite sums yields
\begin{equation}
\mathbf{C}_{p_n}[k]
=
\frac{1}{D}\sum_{r=0}^{D-1}\omega^{r p_n}
\Big(\sum_{d=0}^{D-1}\omega^{-rd}\mathbf{C}_{d}[k]\Big).
\label{eq:Cpn_modes}
\end{equation}

\begin{definition}[Mode-Domain Quantities]
\label{def:Delta_r}
For $r=0,\ldots,D-1$, define the mode correlation matrices
and slow-time mode responses as
\begin{align}
\boldsymbol{\Delta}_r[k]
&\triangleq
\sum_{d=0}^{D-1}\omega^{-rd}\,\mathbf{C}_{d}[k],
\label{eq:Delta_def}\\
S_r(\theta)
&\triangleq
\sum_{n=0}^{N-1} q_n\,\omega^{r p_n}e^{jn\theta}.
\label{eq:Sr_def}
\end{align}
\end{definition}

Substituting \eqref{eq:Cpn_modes} into
\eqref{eq:chi_cpi_discrete} and reordering the finite sums gives
\begin{align}
\boldsymbol{\chi}(k,\theta)
&=
\frac{1}{D}\sum_{r=0}^{D-1}
\Bigg(
\sum_{n=0}^{N-1}
q_n\omega^{r p_n}e^{jn\theta}
\Bigg)
\Bigg(
\sum_{d=0}^{D-1}
\omega^{-rd}\mathbf{C}_d[k]
\Bigg)
\nonumber\\
&=
\frac{1}{D}\sum_{r=0}^{D-1}
S_r(\theta)\,\boldsymbol{\Delta}_r[k].
\label{eq:chi_modes_compact}
\end{align}

From \eqref{eq:Delta_def} and the complementary-set identity
\eqref{eq:comp_set}, the DC mode satisfies
\begin{equation}
\boldsymbol{\Delta}_0[k]
=
DL\,\mathbf{I}_D\,\delta[k].
\label{eq:Delta0_impulse}
\end{equation}
Substituting \eqref{eq:Delta0_impulse} into
\eqref{eq:chi_modes_compact} yields the central mode-domain
decomposition
\begin{equation}
\boldsymbol{\chi}(k,\theta)
=
L\,S_0(\theta)\,\mathbf{I}_D\,\delta[k]
+
\frac{1}{D}\sum_{r=1}^{D-1}
S_r(\theta)\,\boldsymbol{\Delta}_r[k].
\label{eq:chi_split}
\end{equation}

The $r=0$ mode supplies the delay impulse imposed by paraunitarity,
whereas all nonzero-delay range--Doppler coupling is confined to the
$D-1$ nonzero modes. For a fixed waveform library, Doppler-resilient
sidelobe suppression therefore reduces to shaping the slow-time
responses $\{S_r(\theta)\}_{r=1}^{D-1}$.

\begin{remark}[Special Cases]
\label{rem:special_cases}
For $D=2$, \eqref{eq:chi_split} recovers the binary $(P,Q)$
framework of \cite{Dang_arXiv20}. For $D>2$, the design must
suppress $D-1$ nonzero modes simultaneously.
\end{remark}

\subsection{Small-Doppler Behavior and Moment Conditions}

\begin{definition}[Slow-Time Moments]
\label{def:moments}
For each mode $r$ and order $m$, define
\begin{equation}
\mu_{r,m}
\triangleq
\sum_{n=0}^{N-1}q_n\,\omega^{r p_n}n^m,
\qquad
r=0,\ldots,D-1,\quad m=0,1,\ldots
\label{eq:moments_def}
\end{equation}
\end{definition}

Expanding \eqref{eq:Sr_def} about $\theta=0$ and using
\eqref{eq:moments_def} gives
\begin{equation}
S_r(\theta)
=
\sum_{m=0}^{\infty}
\frac{(j\theta)^m}{m!}\,\mu_{r,m}.
\label{eq:Sr_moment_expansion}
\end{equation}

Since \eqref{eq:chi_split} confines all nonzero-delay terms to
the modes $r=1,\ldots,D-1$, we impose the order-$M$
moment-nulling conditions
\begin{equation}
\mu_{r,m}=0,
\qquad
r=1,\ldots,D-1,\quad m=0,\ldots,M.
\label{eq:moment_nulls}
\end{equation}
Under \eqref{eq:moment_nulls}, \eqref{eq:Sr_moment_expansion}
gives
\[
S_r(\theta)=\mathcal{O}(\theta^{M+1}),
\qquad r=1,\ldots,D-1,
\]
so the associated nonzero-delay contributions are locally
suppressed to order $M+1$ in Doppler.

\section{Slow-Time Design}
\label{sec:design}

\subsection{Output SNR}

For a point target with channel response
$\mathbf h\in\mathbb C^D$, the matched-filter output at the origin is
\begin{equation}
\mathbf z(0,0)
=
L(\mathbf 1_N^T\mathbf q)\mathbf h
+
\sum_{n=0}^{N-1}q_n\mathbf w_n(0),
\label{eq:noisy_mainlobe}
\end{equation}
where $\mathbf w_n(0)$ denotes the filtered receiver noise from the
$n$th PRI. Assuming
$\mathbb E\|\mathbf h\|_2^2=D\sigma_b^2$, the aggregate mainlobe
signal power is
\begin{align}
P_s
&=
\mathbb E\!\left[
\left\|
L(\mathbf 1_N^T\mathbf q)\mathbf h
\right\|_2^2
\right]
\nonumber\\
&=
DL^2\sigma_b^2
\left|\mathbf 1_N^T\mathbf q\right|^2.
\label{eq:mainlobe_signal_power}
\end{align}

For temporally white receiver noise with power $N_0$, each transmitted
vector waveform has energy $DL$, and hence
\begin{equation}
\mathbb E\!\left[\|\mathbf w_n(0)\|_2^2\right]
=
N_0DL.
\label{eq:single_pri_noise_power}
\end{equation}
Because the PRIs do not overlap, the filtered noise terms are
uncorrelated across $n$, giving
\begin{equation}
P_w
=
\mathbb E\!\left[
\left\|
\sum_{n=0}^{N-1}q_n\mathbf w_n(0)
\right\|_2^2
\right]
=
N_0DL\|\mathbf q\|_2^2.
\label{eq:output_noise_power}
\end{equation}
The resulting output SNR is
\begin{equation}
\rho
=
\frac{P_s}{P_w}
=
\frac{L\sigma_b^2}{N_0}
\frac{
\left|\mathbf 1_N^T\mathbf q\right|^2
}{
\|\mathbf q\|_2^2
}.
\label{eq:SNR_def}
\end{equation}

Relative to uniform weighting, the SNR efficiency is
\begin{equation}
\eta_{\mathrm{SNR}}
=
\frac{
\left|\mathbf 1_N^T\mathbf q\right|^2
}{
N\|\mathbf q\|_2^2
}
\leq 1,
\label{eq:snr_efficiency}
\end{equation}
where equality holds when $\mathbf q$ is proportional to
$\mathbf 1_N$. Thus, nonuniform weighting provides the degrees of
freedom required for Doppler nulling at the cost of coherent-integration
efficiency.

\subsection{Linear Constraint Representation}

For a fixed schedule $P$, the moment-nulling conditions
\eqref{eq:moment_nulls} form a homogeneous linear system in
$\mathbf q=[\,q_0,\ldots,q_{N-1}\,]^T\in\mathbb C^N$.
For a chosen nulling order $M$, define the $(M+1)\times N$
Vandermonde moment matrix
\begin{equation}
[\mathbf V]_{m,n}=n^m,
\qquad
m=0,\ldots,M,\quad n=0,\ldots,N-1,
\label{eq:V_def}
\end{equation}
and, for each mode $r$, the schedule-dependent modulation matrix
\begin{equation}
\mathbf D_r(P)
\triangleq
\operatorname{diag}\!\left(
\omega^{r p_0},
\omega^{r p_1},
\ldots,
\omega^{r p_{N-1}}
\right).
\label{eq:Dr_def}
\end{equation}
By \eqref{eq:moments_def},
\[
\mathbf V\mathbf D_r(P)\mathbf q
=
[\,\mu_{r,0},\ldots,\mu_{r,M}\,]^T.
\]
Stacking the moment constraints for all nonzero modes therefore gives
the global system
\begin{equation}
\mathbf A(P)\mathbf q=\mathbf 0,
\qquad
\mathbf A(P)
\triangleq
\begin{bmatrix}
\mathbf V\mathbf D_1(P)\\
\vdots\\
\mathbf V\mathbf D_{D-1}(P)
\end{bmatrix}.
\label{eq:A_stack_def}
\end{equation}
Each block row of $\mathbf A(P)$ enforces the $M+1$ moment
constraints associated with one nonzero mode.

\subsection{SNR-Optimal Weights for a Fixed Schedule}

For a fixed schedule $P$, let
$\mathcal S=\mathcal N(\mathbf A(P))\subseteq\mathbb C^N$
denote the feasible subspace. Using \eqref{eq:SNR_def}, the
SNR-optimal weight design is equivalent to
\begin{equation}
\max_{\mathbf q\in\mathcal S,\;\mathbf q\neq\mathbf 0}
\frac{|\mathbf 1_N^T\mathbf q|^2}{\|\mathbf q\|_2^2}.
\label{eq:q_opt_fixedP}
\end{equation}

\begin{theorem}[SNR-Optimal Weight]
\label{thm:projection_opt}
Let $\mathbf P_{\mathcal S}$ denote the orthogonal projector onto
$\mathcal S$. Provided that
$\mathbf P_{\mathcal S}\mathbf 1_N\neq\mathbf 0$, a unit-norm
maximizer of \eqref{eq:q_opt_fixedP} is
\begin{equation}
\mathbf q^\star(P)
=
\frac{\mathbf P_{\mathcal S}\mathbf 1_N}
{\|\mathbf P_{\mathcal S}\mathbf 1_N\|_2}.
\label{eq:q_star_proj}
\end{equation}
The corresponding maximum output SNR is
\begin{equation}
\rho_{\max}(P)
=
\frac{L\sigma_b^2}{N_0}
\|\mathbf P_{\mathcal S}\mathbf 1_N\|_2^2.
\label{eq:SNR_max_proj}
\end{equation}
\end{theorem}

\begin{proof}
For any $\mathbf q\in\mathcal S$,
\[
|\mathbf 1_N^T\mathbf q|
=
|(\mathbf P_{\mathcal S}\mathbf 1_N)^H\mathbf q|
\leq
\|\mathbf P_{\mathcal S}\mathbf 1_N\|_2\|\mathbf q\|_2,
\]
with equality when $\mathbf q$ is proportional to
$\mathbf P_{\mathcal S}\mathbf 1_N$.
\end{proof}

\subsection{Design of the Column Scheduling Sequence}

Using the fixed-schedule solution \eqref{eq:q_star_proj} in
\eqref{eq:SNR_max_proj}, the remaining design problem is
\begin{equation}
P^\star
\in
\arg\max_{P\in\{0,\ldots,D-1\}^N}
\left\|
\mathbf P_{\mathcal N(\mathbf A(P))}\mathbf 1_N
\right\|_2^2,
\label{eq:P_opt_discrete}
\end{equation}
subject to $\mathbf A(P)$ having a nontrivial null space.

The schedule search in \eqref{eq:P_opt_discrete} contains $D^N$
candidates, motivating the following heuristic.

\begin{remark}[Numerical Schedule Search]
\label{rem:schedule_search}
The schedules used in Section~\ref{sec:results} were obtained by a
genetic search restricted to balanced $D$-ary sequences. The initial population
contained the periodic schedule together with randomly generated
balanced schedules. Each candidate was scored using the projection objective in
\eqref{eq:P_opt_discrete}. Before assembling the constraints,
the Vandermonde rows were orthonormalized using a QR
decomposition to improve numerical conditioning without
changing the null space. The objective was then evaluated
by an SVD of the real-valued constraint matrix described in
Section~\ref{sec:real_form}. At each generation, high-scoring schedules were retained, and new candidates were generated by crossover and mutation. The column counts of each new candidate were adjusted as needed to preserve balance. The search stopped when the best objective value ceased to improve. The
SNR-optimal weights for the resulting schedule were then computed
using \eqref{eq:q_star_proj}.
\end{remark}

\section{Numerical Illustrations}
\label{sec:results}

All waveform-level illustrations use the BTQ/MSK realization of
Section~\ref{sec:msk_realization} implemented according to
\eqref{eq:btq_waveform}. Unless otherwise stated, the simulations use
$N=64$, $L=128$ biphase coefficients (64 paired-chip intervals),
$N_s=64$ samples per interval, and 513 uniformly spaced Doppler samples over
$[-\pi,\pi]$. Correlations are evaluated at the full sample rate, with delay
reported in chip intervals.

The aggregate ambiguity matrix $\boldsymbol{\chi}(k,\theta)$ in
\eqref{eq:chi_cpi_discrete} contains $D^2$ matched-filter responses. We denote
by $\chi_{i,j}(k,\theta)$ the response of the $j$th matched filter to the
waveform transmitted on channel $i$, after weighted combination across the
$N$ PRIs. For the scalar ambiguity plots, these responses are combined
coherently as
\begin{equation}
\chi(k,\theta)
\triangleq
\sum_{i=1}^{D}\sum_{j=1}^{D}\chi_{i,j}(k,\theta).
\label{eq:scalar_caf_sum}
\end{equation}
The plots show $20\log_{10}(|\chi(k,\theta)|/|\chi(0,0)|)$ with a
$-110$~dB display floor.

\begin{figure*}[!h]
\centering

\begin{tabular}{@{}cc@{}}
\includegraphics[width=.45\textwidth]{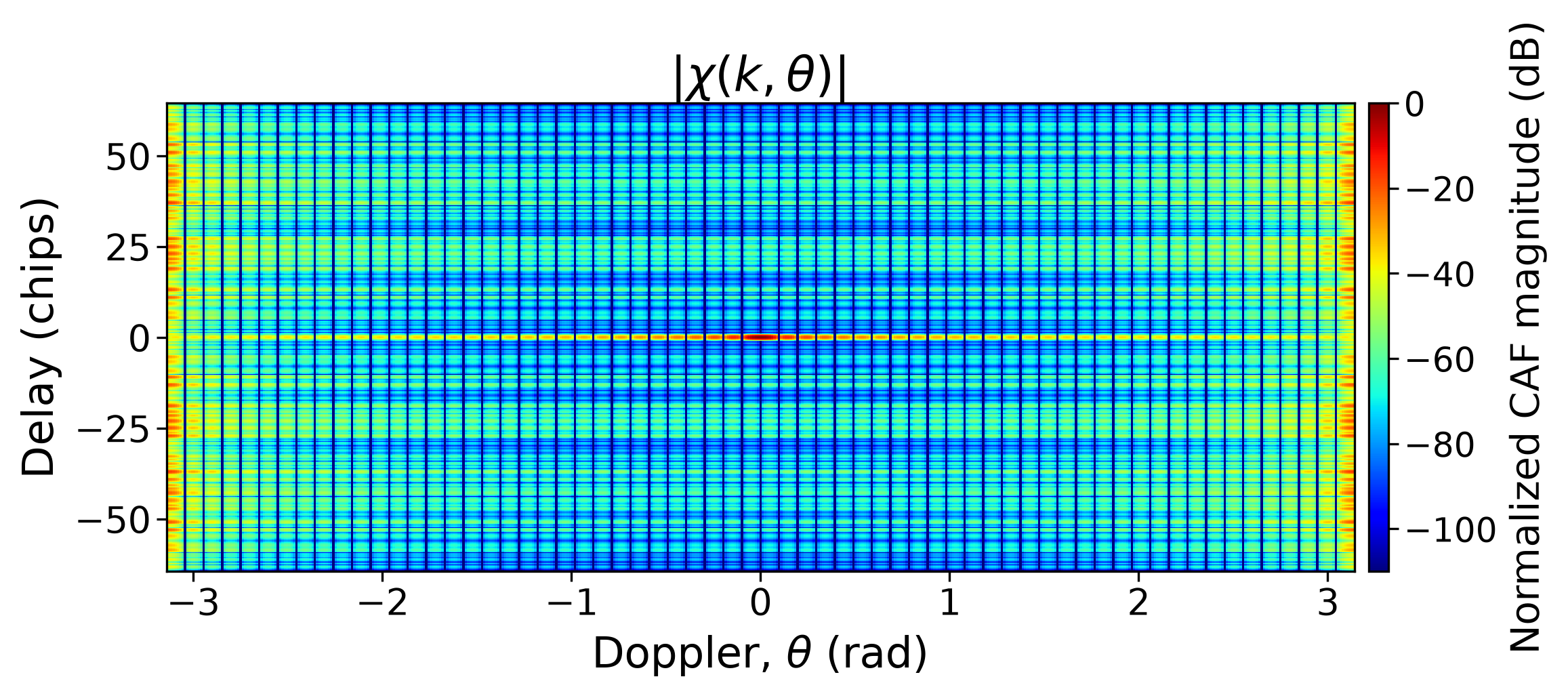} &
\includegraphics[width=.45\textwidth]{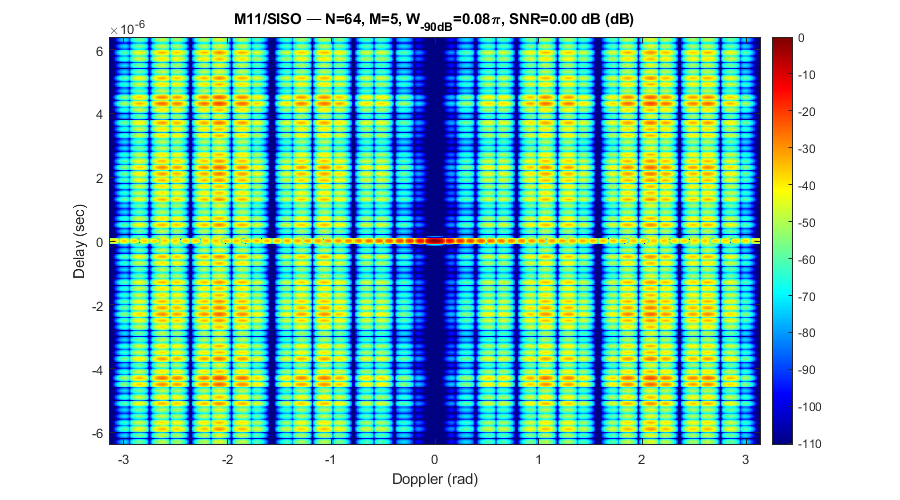} \\[-1pt]
{\small (a) Alternating ($M=0$, SNR $=0.000$ dB)} &
{\small (b) PTM ($M=5$, SNR $=0.000$ dB)} \\[1pt]
\includegraphics[width=.45\textwidth]{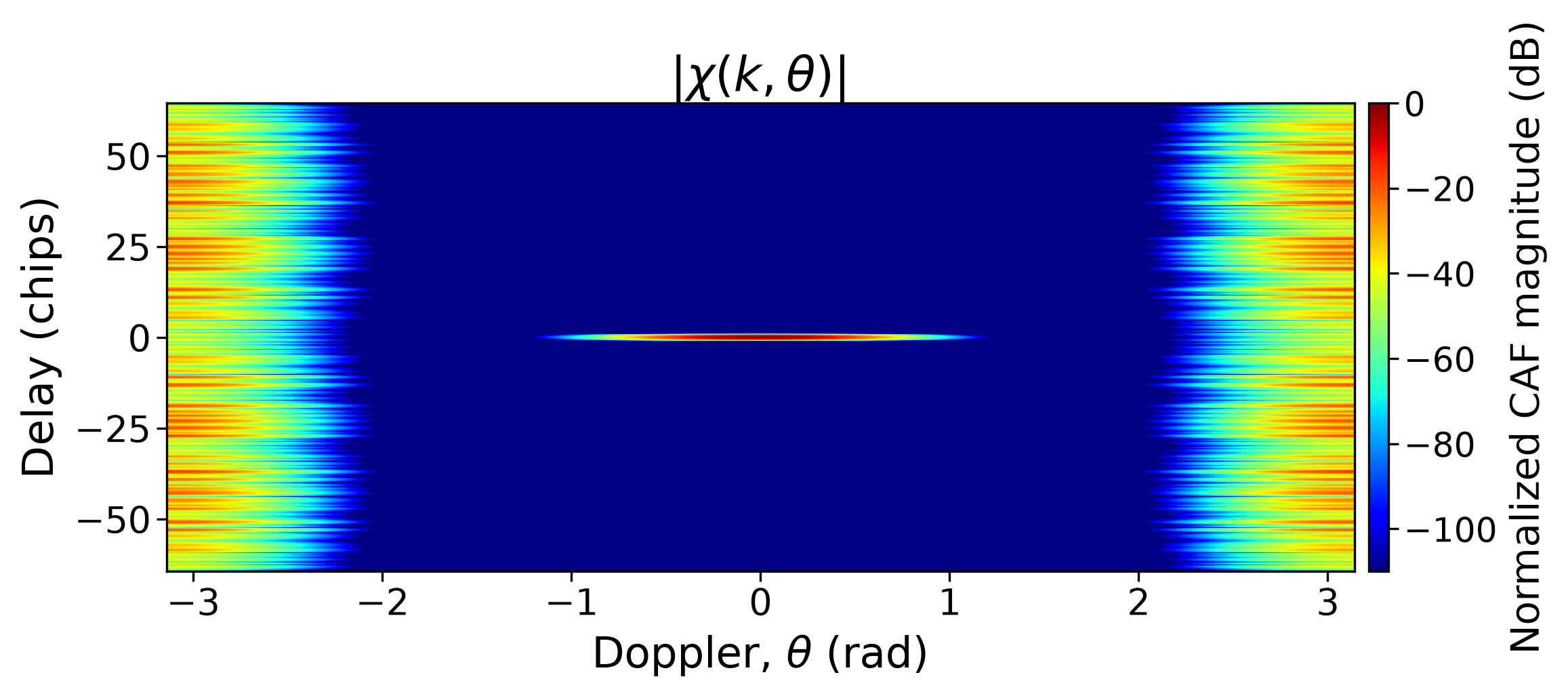} &
\includegraphics[width=.45\textwidth]{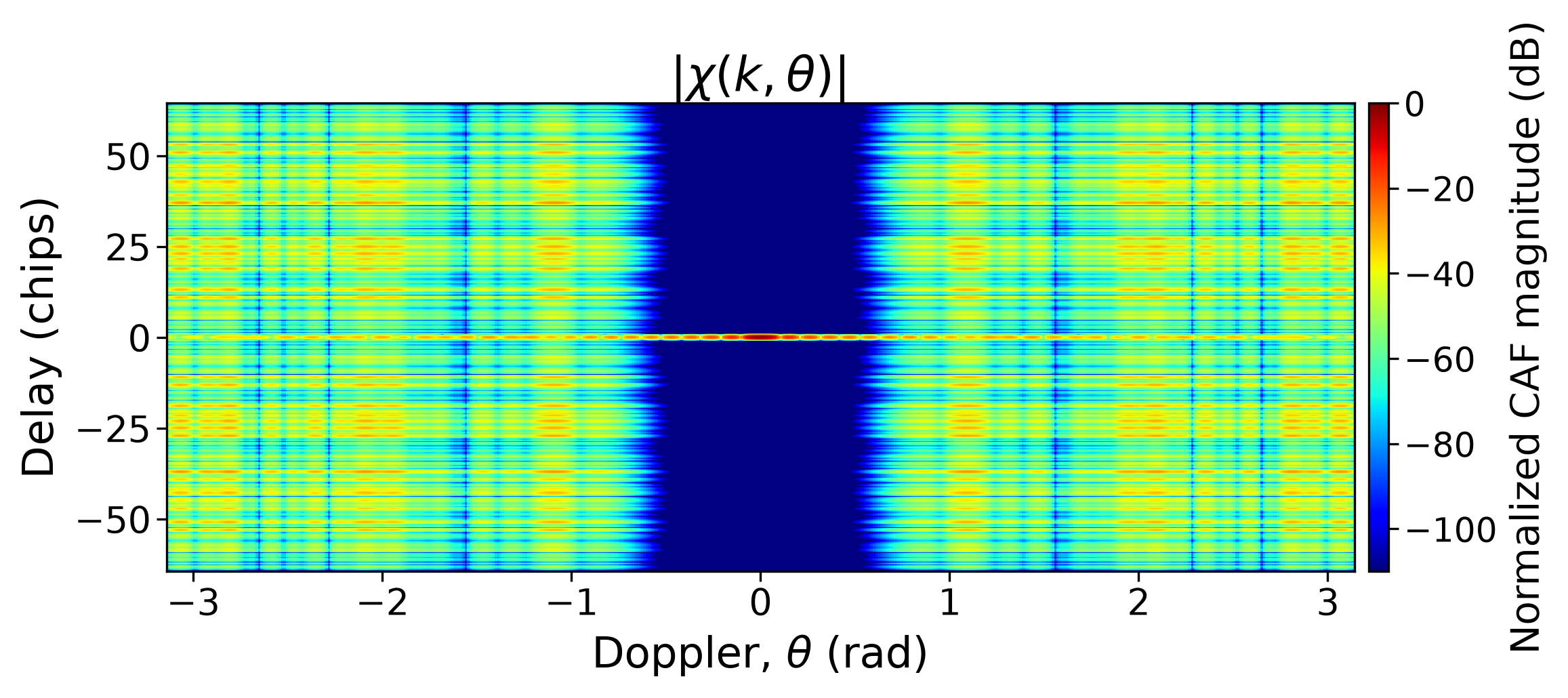} \\[-1pt]
{\small (c) Binomial ($M=62$, SNR $=-6.571$ dB)} &
{\small (d) Proposed ($M=25$, SNR $=-0.339$ dB)}
\end{tabular}
\vspace{-1mm}
\caption{Normalized CPI ambiguity magnitude $|\chi(k,\theta)|$ for the
binary case ($D=2$, $N=64$): (a) alternating, (b) PTM, (c) binomial,
and (d) the proposed design.}
\label{fig:siso_compare}
\end{figure*}

\begin{figure*}[!h]
\centering
\includegraphics[width=.85\textwidth]{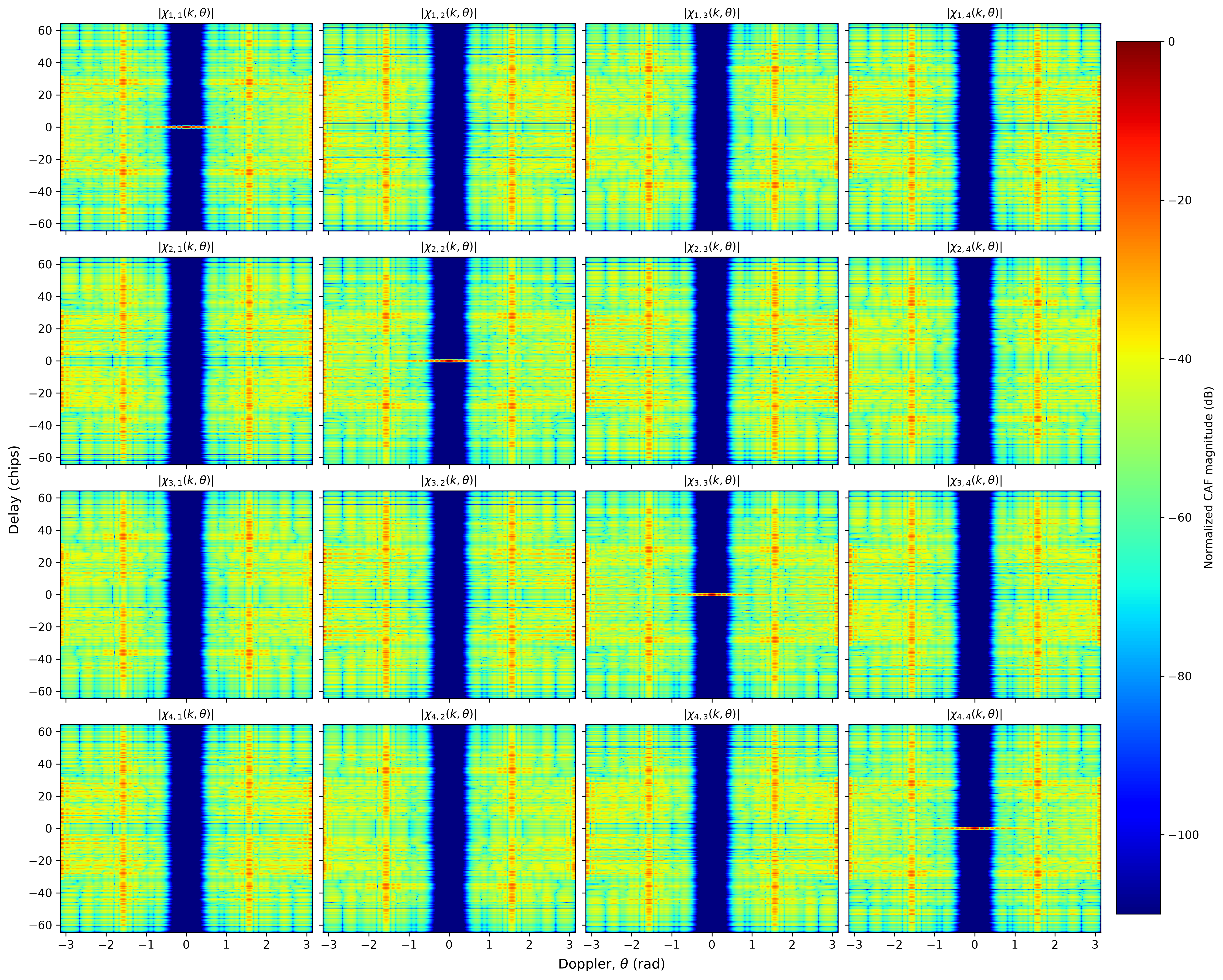}
\vspace{-1mm}
\caption{Normalized CPI ambiguity-matrix magnitudes for the $D=4$,
$N=64$, and $M=16$ design, with Doppler-clearance width $0.16\pi$
and SNR loss $-0.969$~dB.}
\label{fig:4x4_caf}
\end{figure*}

\subsection{Binary Complementary Case ($D=2$)}
\label{sec:siso_results}

Figure~\ref{fig:siso_compare} compares several representative $(P,Q)$
constructions, including the classical alternating schedule,
the PTM design and the binomial
weighting scheme, together with the proposed maximum-SNR design obtained
from the optimization framework in Section~\ref{sec:design}.

Unlike classical constructions, which correspond to fixed points along
this spectrum, the proposed design framework allows $M$ to be selected
according to the desired Doppler-clearance width, providing a flexible
SNR–robustness tradeoff.

\subsection{Four-Waveform Case ($D=4$)}
\label{sec:4x4_results}

We next consider the $D=4$ case, in which four complementary waveforms
are transmitted simultaneously. Figure~\ref{fig:4x4_caf} shows the
magnitudes $|\chi_{i,j}(k,\theta)|$ of all 16 matched-filter responses for
$N=64$ and null order $M=16$. The resulting slow-time design has an SNR
loss of $-0.969$~dB.

The same Doppler-clearance region appears in the auto-correlation responses
$\chi_{i,i}(k,\theta)$ and the cross-correlation responses
$\chi_{i,j}(k,\theta)$, $i\neq j$. From \eqref{eq:chi_split}, the common
modal spectra are controlled by the slow-time design, whereas the
corresponding mode-matrix entries determine the delay-dependent structure
of each response. Consequently, the imposed Doppler null order is inherited
by every ambiguity-matrix entry, although their sidelobe structures outside
the clearance region differ.
\begin{figure}[!h]
\centering
\includegraphics[width=\columnwidth]{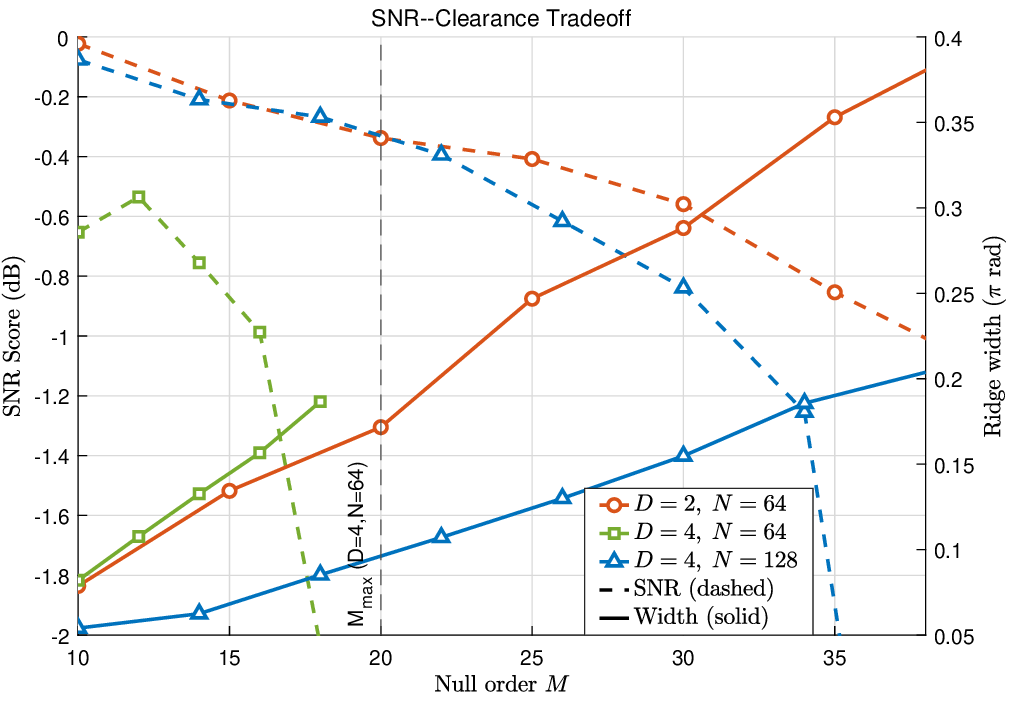}
\caption{Mainlobe SNR \eqref{eq:SNR_def} and $-90$~dB Doppler-clearance width
versus null order $M$ for $(D,N)=(2,64)$, $(4,64)$, and $(4,128)$.
}
\label{fig:sweep_tradeoff}
\end{figure}

\subsection{CPI Length, Waveform Dimension, and Doppler-Clearance Tradeoff}
\label{sec:sweep_tradeoff}

Several trends are evident. Increasing the null order $M$
widens the Doppler-clearance region but reduces coherent SNR due to
the additional slow-time constraints imposed on the waveform weights.
For a fixed CPI length $N$, increasing the waveform dimension $D$
reduces the maximum achievable null order since more sidelobe modes
must be suppressed simultaneously.



\begin{figure}[!h]
\centering
\includegraphics[width=.45\textwidth]{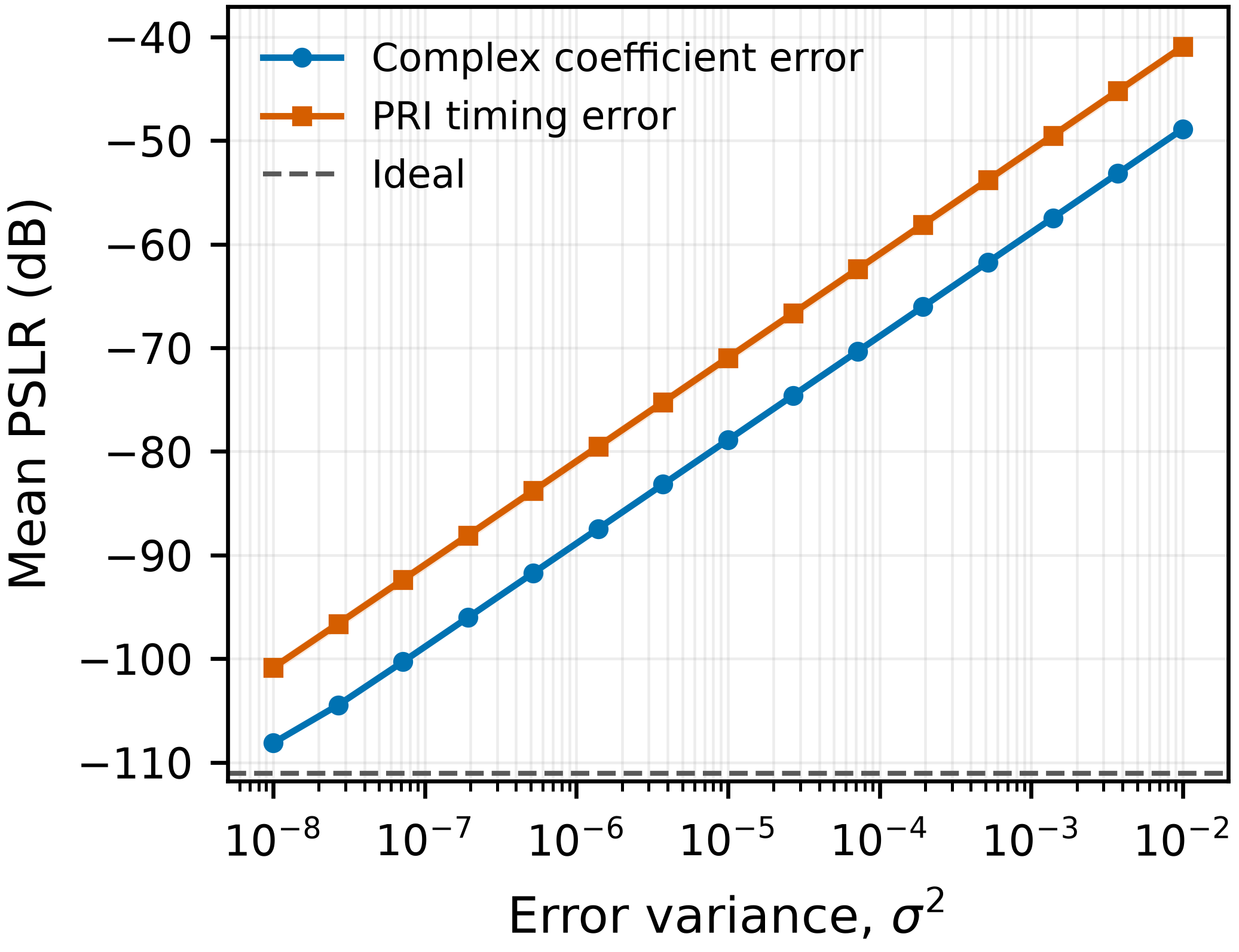}
\caption{Sensitivity to waveform-implementation errors: mean PSLR within $|\theta|\leq0.1\pi$ versus complex coefficient-error
variance and normalized PRI timing-error variance.}
\label{fig:implementation_sensitivity_a}
\end{figure}

\begin{figure}[!h]
\centering
\includegraphics[width=.45\textwidth]{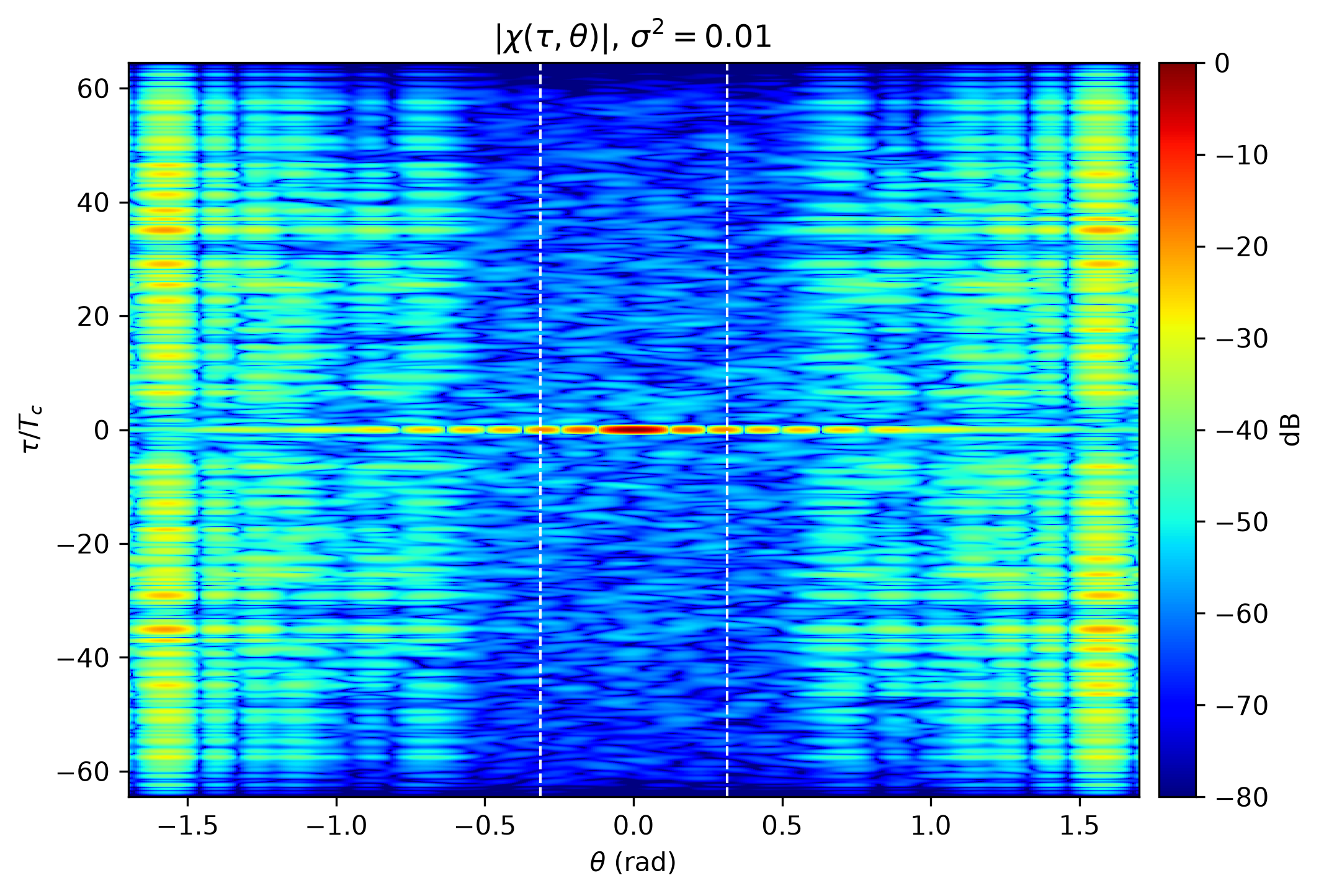}
\caption{Sensitivity to waveform-implementation errors:  normalized CAF magnitude for a
representative coefficient-error realization at $\sigma^2=10^{-2}$}\label{fig:implementation_sensitivity_b}
\end{figure}


\subsection{Sensitivity to Implementation Errors}
\label{sec:implementation_sensitivity}

The preceding numerical results assume exact realization and
synchronization of the designed waveform. Practical transmitters exhibit
amplitude, phase, and timing variations across channels, PRIs, and chips
due to power-amplifier drift, thermal variation, power-supply
fluctuations, local-oscillator instability, and converter errors \cite{ShraderGregersHansen90}. 
The resulting degradation depends on both
the perturbation magnitude and its correlation across the CPI. 

We therefore consider two representative perturbation models to illustrate
how such errors fill the designed zero-Doppler clearance region.
The simulations use $M=16$ design used in D=4 case.
First, independent complex coefficient errors are applied according to
\begin{equation}
\widetilde S_{n,i}[\ell]
=
\left(1+\sqrt{\sigma^2}\,z_{n,i}[\ell]\right)
S_{i,p_n}[\ell],
\label{eq:coefficient_perturbation}
\end{equation}
where $z_{n,i}[\ell]\sim\mathcal{CN}(0,1)$ denotes a
zero-mean, unit-variance circularly symmetric complex Gaussian random
variable, independent across PRIs, transmit channels, and chips.

Second, a common start-time offset is applied to all transmitters within
each PRI. The resulting CPI waveform is
\begin{equation}
\widetilde x_i^{(\tau)}(t)
=
\sum_{n=0}^{N-1}
x_{i,p_n}^{\mathrm{BTQ}}
\left(t-nT-\delta_n\right),
\qquad
\delta_n=T_c\sqrt{\sigma_\tau^2}\,u_n,
\label{eq:pri_timing_perturbation}
\end{equation}
where $\{u_n\}$ is obtained from independent real Gaussian draws by
removing their sample mean and normalizing their sample variance to
unity. Thus, $\sigma_\tau^2$ is the timing-error variance normalized by
$T_c^2$. Fractional offsets are evaluated on the oversampled correlation
grid. In both models, the receiver uses the nominal waveform as its
matched-filter reference.

For each value of $\sigma^2$ or $\sigma_\tau^2$, the PSLR is evaluated
over $|\theta|\leq0.1\pi$, excluding the delay mainlobe. The reported
statistic is
\begin{equation}
\overline{\mathrm{PSLR}}_{\mathrm{dB}}
=
20\log_{10}
\left[
\frac{1}{50}
\sum_{r=1}^{50}
\max_{(k,\theta)\in\mathcal S}
\frac{
\left|\widetilde\chi_r(k,\theta)\right|
}{
\left|\widetilde\chi_r(0,0)\right|
}
\right],
\label{eq:mean_sensitivity_pslr}
\end{equation}
where $\mathcal S$ contains the nonmainlobe delays and Doppler
frequencies satisfying $|\theta|\leq0.1\pi$. Thus, the linear-magnitude
PSLRs are averaged before conversion to decibels.

Figure~\ref{fig:implementation_sensitivity_a} shows the mean PSLR as each
normalized error variance is swept from $10^{-8}$ to $10^{-1}$. Both
perturbations progressively raise the sidelobe level within the clearance
region. Figure~\ref{fig:implementation_sensitivity_b} shows a representative
coefficient-error realization at $\sigma^2=10^{-2}$ and illustrates the
corresponding delay--Doppler structure of the residual sidelobes.

\section{Feasibility and Degrees of Freedom}
\label{sec:feasibility}

\subsection{CPI Length Requirements}

From \eqref{eq:A_stack_def}, $
\mathbf A(P)\in
\mathbb C^{(D-1)(M+1)\times N}$ and 
$\dim\mathcal N(\mathbf A(P))
=
N-\operatorname{rank}\mathbf A(P)$. 
Thus, the exact condition for nonzero Doppler-nulling weights is
$\operatorname{rank}\mathbf A(P)<N$. Since
$\operatorname{rank}\mathbf A(P)\leq(D-1)(M+1)$,
\begin{equation}
N>(D-1)(M+1)
\label{eq:feasibility_condition}
\end{equation}
is a schedule-independent sufficient condition. When
$\mathbf A(P)$ attains its maximum possible rank,
\eqref{eq:feasibility_condition} is also necessary and gives the
generic feasibility threshold. Each additional nonzero mode
contributes $M+1$ constraint rows, while each increase in null order
adds $D-1$ rows, reducing the available slow-time degrees of freedom
unless these rows are linearly dependent.

\subsection{Real-Valued Optimality and Computation}
\label{sec:real_form}

Although $\mathbf A(P)$ is generally complex, its conjugate-mode
structure ensures that the SNR-optimal weights are real. Define the
constraint block
$\mathbf A_r(P)\triangleq\mathbf V\mathbf D_r(P)$. And,
\[
\omega^{(D-r)p_n}
=
e^{j2\pi p_n}e^{-j2\pi r p_n/D}
=
e^{-j2\pi r p_n/D}
=
\overline{\omega^{r p_n}}.
\]
Together with the fact that $\mathbf V$ is real, this gives
\begin{equation}
\mathbf A_{D-r}(P)
=
\overline{\mathbf A_r(P)}.
\label{eq:Ar_conj_pair_short}
\end{equation}
Consequently, if
$\mathbf q\in\mathcal N(\mathbf A(P))$, then
\[
\mathbf A_r(P)\overline{\mathbf q}
=
\overline{\mathbf A_{D-r}(P)\mathbf q}
=
\mathbf 0
\]
for every nonzero mode $r$. Thus,
$\mathcal N(\mathbf A(P))$ is closed under conjugation.

Let $\mathcal S=\mathcal N(\mathbf A(P))$ and
$\widehat{\mathbf q}=\mathbf P_{\mathcal S}\mathbf 1_N$.
Because $\mathbf 1_N$ is real and $\mathcal S$ is closed under
conjugation, $\overline{\widehat{\mathbf q}}$ is also an orthogonal
projection of $\mathbf 1_N$ onto $\mathcal S$. Uniqueness of the
orthogonal projection therefore gives
$\widehat{\mathbf q}=\overline{\widehat{\mathbf q}}$, so
$\widehat{\mathbf q}\in\mathbb R^N$. Hence, allowing complex weights
provides no improvement in the SNR objective \eqref{eq:SNR_def}.

The optimal weights may therefore be computed using real arithmetic.
For $\mathbf q\in\mathbb R^N$,
$\mathbf A(P)\mathbf q=\mathbf 0$ is equivalent to
\begin{equation}
\begin{bmatrix}
\Re\{\mathbf A(P)\}\\
\Im\{\mathbf A(P)\}
\end{bmatrix}
\mathbf q
=
\mathbf 0.
\label{eq:real_embed_restricted}
\end{equation}
Moreover, \eqref{eq:Ar_conj_pair_short} makes the constraints for
$r$ and $D-r$ redundant. It therefore suffices to retain one block
from each conjugate pair, separating its real and imaginary parts.
When $D$ is even, the self-conjugate mode $r=D/2$ is real and is
retained only once. The resulting real system has
$(D-1)(M+1)$ rows, so the feasibility condition
\eqref{eq:feasibility_condition} is unchanged.

\subsection{CPI Length $N$, Null Order $M$, and Doppler Clearance Width}

We next relate the CPI length and null order to the Doppler interval over
which the range sidelobes remain suppressed. For a fixed paraunitary
waveform library, define the normalized peak sidelobe level
\begin{equation}
\mathcal{L}(\theta)
\triangleq
\frac{
\displaystyle\max_{k\neq 0}
\left|\boldsymbol{\chi}(k,\theta)\right|_2
}{
\left|\boldsymbol{\chi}(0,0)\right|_2
},
\label{eq:normalized_sidelobe_level}
\end{equation}
where $|\cdot|_2$ denotes the matrix spectral norm. For a prescribed
threshold $\varepsilon$, the Doppler-clearance half-width $\theta_w$ is
the largest value for which
$\mathcal{L}(\theta)\leq\varepsilon$ throughout
$|\theta|\leq\theta_w$.

Under \eqref{eq:moment_nulls}, the zeroth-order constraints give
$S_r(0)=0$ for $r=1,\ldots,D-1$. Hence, from
\eqref{eq:chi_split},
\begin{equation}
\boldsymbol{\chi}(0,0)
=
L(\mathbf 1_N^T\mathbf q)\mathbf I_D,
\qquad
\left\|\boldsymbol{\chi}(0,0)\right\|_2
=
L\left|\mathbf 1_N^T\mathbf q\right|.
\label{eq:mainlobe_zero}
\end{equation}

Define the waveform-dependent constant
\begin{equation}
\Gamma_{\Delta}
\triangleq
\max_{\substack{k\neq 0\\1\leq r\leq D-1}}
\left\|\boldsymbol{\Delta}_r[k]\right\|_2.
\label{eq:Gamma_Delta}
\end{equation}
For $k\neq0$, \eqref{eq:chi_split} and the triangle inequality give
\begin{equation}
\max\limits_{k\neq 0}
\left\|\boldsymbol{\chi}(k,\theta)\right\|_2
\leq
\frac{\Gamma_{\Delta}}{D}
\sum_{r=1}^{D-1}\left|S_r(\theta)\right|.
\label{eq:matrix_sidelobe_modal_bound}
\end{equation}
Consequently,
\begin{equation}
\mathcal{L}(\theta)
\leq
\frac{\Gamma_{\Delta}}{DL}
\frac{
\displaystyle\sum_{r=1}^{D-1}|S_r(\theta)|
}{
|\mathbf 1_N^T\mathbf q|
}.
\label{eq:normalized_sidelobe_modal_bound}
\end{equation}
For a fixed waveform library, $\Gamma_{\Delta}$ is independent of
$N$ and $M$. Their effect on the clearance region is therefore
controlled by the normalized modal sum in
\eqref{eq:normalized_sidelobe_modal_bound}.

\begin{theorem}[Modal Growth Bound]
Assume that $M\geq0$, $\mathbf 1_N^T\mathbf q\neq0$, and the
moment-nulling conditions \eqref{eq:moment_nulls} hold. Then, for every
$\theta$,
\begin{equation}
\frac{
\displaystyle\sum_{r=1}^{D-1}|S_r(\theta)|
}{
|\mathbf 1_N^T\mathbf q|
}
\leq
\frac{D-1}{\sqrt{\eta_{\mathrm{SNR}}}}
\sum_{m=M+1}^{\infty}
\frac{(N|\theta|)^m}{m!\sqrt{2m+1}}.
\label{eq:NM_tail_bound}
\end{equation}
\end{theorem}

\begin{proof}
From \eqref{eq:Sr_moment_expansion} and
\eqref{eq:moment_nulls},
\begin{equation}
\sum_{r=1}^{D-1}|S_r(\theta)|
\leq
\sum_{m=M+1}^{\infty}
\frac{|\theta|^m}{m!}
\sum_{r=1}^{D-1}|\mu_{r,m}|.
\label{eq:modal_tail_initial}
\end{equation}
For $m\geq M+1\geq1$, the power sum satisfies
\begin{align}
\sum_{n=0}^{N-1}n^{2m}
&\leq
\int_0^{N-1}(x+1)^{2m}\,dx
\nonumber\\
&=
\frac{N^{2m+1}-1}{2m+1}
\leq
\frac{N^{2m+1}}{2m+1}.
\label{eq:power_sum_bound}
\end{align}
Cauchy--Schwarz therefore gives
\begin{align}
|\mu_{r,m}|
&\leq
\|\mathbf q\|_2
\left(\sum_{n=0}^{N-1}n^{2m}\right)^{1/2}
\nonumber\\
&\leq
\frac{N^{m+1/2}}{\sqrt{2m+1}}
\|\mathbf q\|_2.
\label{eq:mu_CS_bound}
\end{align}
Substitution into \eqref{eq:modal_tail_initial} yields
\begin{align}
\sum_{r=1}^{D-1}|S_r(\theta)|
&\leq
(D-1)\sqrt{N}\,\|\mathbf q\|_2
\nonumber\\
&\quad{}\times
\sum_{m=M+1}^{\infty}
\frac{(N|\theta|)^m}
{m!\sqrt{2m+1}}.
\label{eq:Sr_sum_bound}
\end{align}
Finally, \eqref{eq:snr_efficiency} gives
\begin{equation}
\frac{\sqrt{N}\,\|\mathbf q\|_2}
{|\mathbf 1_N^T\mathbf q|}
=
\frac{1}{\sqrt{\eta_{\mathrm{SNR}}}},
\end{equation}
which proves \eqref{eq:NM_tail_bound}.
\end{proof}

Combining \eqref{eq:normalized_sidelobe_modal_bound} and
\eqref{eq:NM_tail_bound} gives the explicit sidelobe bound
\begin{equation}
\mathcal{L}(\theta)
\leq
\frac{\Gamma_{\Delta}(D-1)}
{DL\sqrt{\eta_{\mathrm{SNR}}}}
\sum_{m=M+1}^{\infty}
\frac{(N|\theta|)^m}{m!\sqrt{2m+1}}.
\label{eq:normalized_sidelobe_growth_bound}
\end{equation}

\begin{corollary}[CPI-Length Scaling]
\label{cor:MN_scaling}
For fixed $D$, a fixed waveform library, and SNR efficiency
bounded away from zero, a null order growing linearly with
$N\theta_w$ suffices to maintain a fixed positive tolerance
in \eqref{eq:normalized_sidelobe_growth_bound} throughout
$|\theta|\leq\theta_w$. Specifically, a sufficient null order
can be chosen with
\begin{equation}
M=eN\theta_w+\mathcal O(1).
\label{eq:MN_width_scaling}
\end{equation}
\end{corollary}

\begin{proof}
For $m\geq M+1$, we have $e^{m-(M+1)}\geq1$, so
\[
\frac{(N|\theta|)^m}{m!}
\leq
e^{m-(M+1)}\frac{(N|\theta|)^m}{m!}
=
e^{-(M+1)}\frac{(eN|\theta|)^m}{m!}.
\]
Consequently,
\begin{align*}
\sum_{m=M+1}^{\infty}
\frac{(N|\theta|)^m}{m!\sqrt{2m+1}}
&\leq
\sum_{m=M+1}^{\infty}\frac{(N|\theta|)^m}{m!}\\
&\leq
e^{-(M+1)}
\sum_{m=M+1}^{\infty}\frac{(eN|\theta|)^m}{m!}\\
&\leq
e^{-(M+1)}
\sum_{m=0}^{\infty}\frac{(eN|\theta|)^m}{m!}\\
&=
\exp\!\bigl(eN|\theta|-(M+1)\bigr).
\end{align*}
Substituting into
\eqref{eq:normalized_sidelobe_growth_bound}, a tolerance
$\varepsilon>0$ is therefore guaranteed throughout
$|\theta|\leq\theta_w$ whenever
\[
M+1\geq eN\theta_w+
\log\!\left(
\frac{\Gamma_{\Delta}(D-1)}
{\varepsilon DL\sqrt{\eta_{\mathrm{SNR}}}}
\right).
\]
The logarithmic term remains bounded under the stated
assumptions. Choosing the smallest nonnegative integer $M$
satisfying this inequality gives \eqref{eq:MN_width_scaling}.
\end{proof}

\begin{figure}[!th]
\centering
\includegraphics[width=0.85\columnwidth]{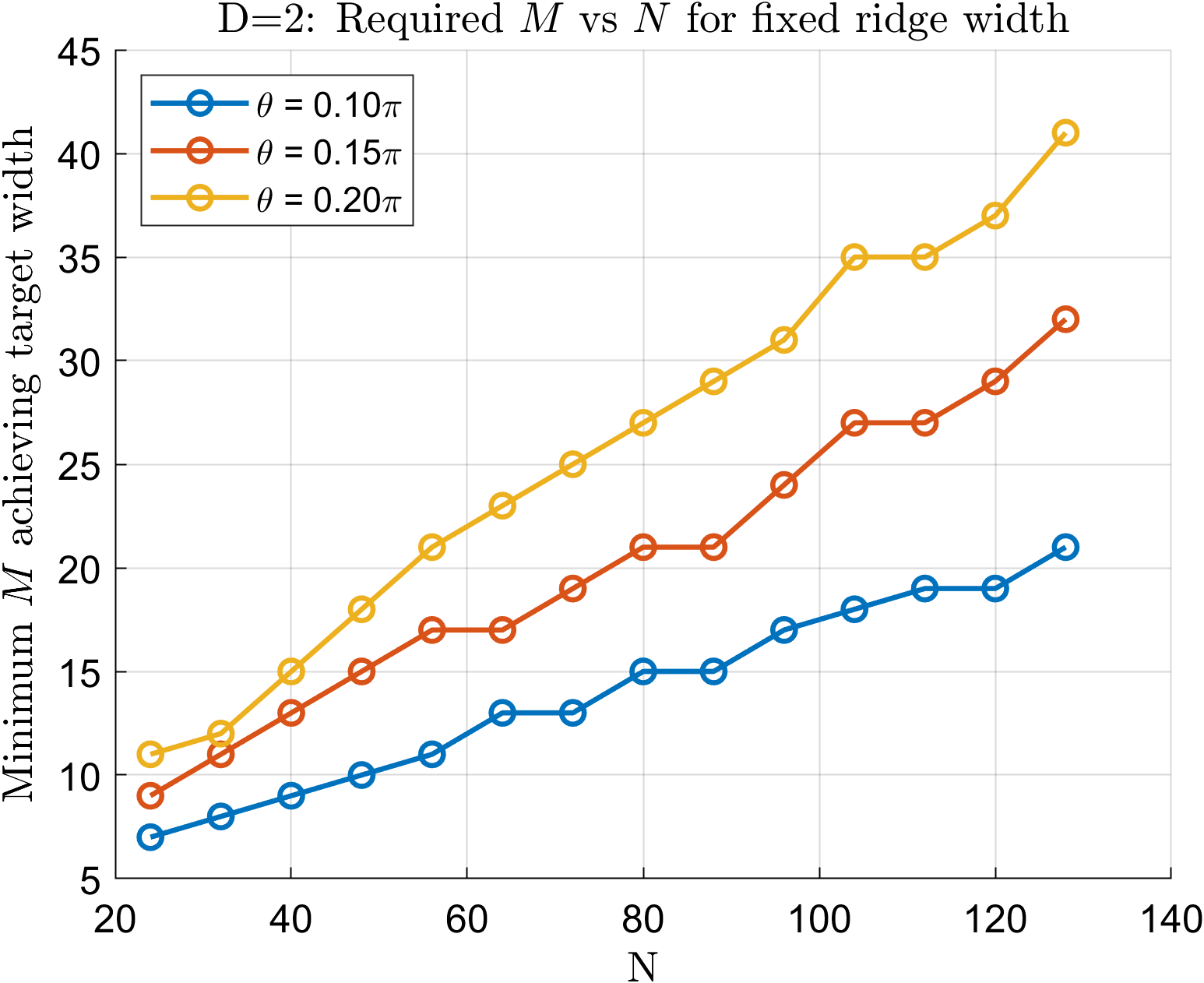}
\caption{Minimum null order required to maintain prescribed
Doppler-clearance half-widths
$\theta_w\in{0.10\pi,0.15\pi,0.20\pi}$ as the CPI length $N$
increases. Each point is obtained from CAF simulation using the
smallest $M$ satisfying the prescribed clearance criterion. For each
fixed $\theta_w$, the approximately linear growth supports
$M_{\min}(N;\theta_w)\approx c(\theta_w)N$; the different slopes
reflect the dependence of $c(\theta_w)$ on the desired clearance
width.}
\label{fig:MN_scaling}
\end{figure}

\begin{corollary}[Waveform-Dimension Tradeoff]
Suppose that $\mathbf A(P)$ attains its maximum possible rank.
A nontrivial null space then requires
\begin{equation}
M+1<\frac{N}{D-1},
\label{eq:M_dimension_limit}
\end{equation}
and hence
\begin{equation}
\frac{M}{N}
=
\mathcal O\!\left(\frac{1}{D-1}\right).
\label{eq:M_dimension_scaling}
\end{equation}
Combining \eqref{eq:M_dimension_limit} with
\eqref{eq:MN_width_scaling}, the clearance half-width
certified by Corollary~\ref{cor:MN_scaling} satisfies
\begin{equation}
\theta_w
=
\mathcal O\!\left(\frac{1}{D-1}\right),
\label{eq:dimension_width_scaling}
\end{equation}
provided the additive term in \eqref{eq:MN_width_scaling}
remains uniformly bounded as $D$ varies.
\end{corollary}

\section{Conclusion}

This paper developed a framework for Doppler-resilient coordination
of paraunitary MIMO waveform libraries. The modal decomposition
reduces prescribed-order Doppler nulling to linear constraints on
the slow-time receive weights. For a fixed schedule, subspace
projection gives the SNR-optimal weights, leaving a discrete
schedule search to complete the design.

The analysis relates CPI length, null order, and waveform dimension
to the available degrees of freedom and a sufficient bound on
Doppler clearance. Numerical examples demonstrate sidelobe
suppression across the MIMO ambiguity matrix and illustrate the
clearance--SNR tradeoff. These results enable Doppler resilience
through slow-time coordination of a fixed complementary waveform
library.

\appendices

\section{Lifting-Based Paraunitary Waveform Library Construction}
\label{app:PUconstruction}

The biphase libraries used in the numerical examples are generated
using a specialization of the paraunitary construction in
\cite{BudiSinSpasojevic}. Initialize
$\mathbf S^{(0)}(z)=\mathbf U^{(0)}$, where $\mathbf U^{(0)}$
is an unnormalized Sylvester Hadamard matrix satisfying
$\mathbf U^{(0)}\mathbf U^{(0)H}=D\mathbf I_D$.

At stage $s$, partition the rows into $D/2$ disjoint pairs.
Let $\mathbf P^{(s)}$ place the paired rows in adjacent positions,
and assign each pair a sign $\sigma_{s,r}\in\{-1,1\}$. Define
\begin{equation}
\mathbf T_{2,r}^{(s)}(z)
=
\begin{bmatrix}
1 & \sigma_{s,r}z^{-1}\\
1 & -\sigma_{s,r}z^{-1}
\end{bmatrix},
\label{eq:pairwise_lifting_factor}
\end{equation}
and assemble the stage matrix as
\begin{equation}
\mathbf T^{(s)}(z)
=
\mathbf P^{(s)T}
\operatorname{blkdiag}\!\left(
\mathbf T_{2,1}^{(s)}(z),\ldots,
\mathbf T_{2,D/2}^{(s)}(z)
\right)
\mathbf P^{(s)}.
\label{eq:stage_lifting_matrix}
\end{equation}
The recursion is
\begin{equation}
\mathbf S^{(s)}(z)
=
\mathbf T^{(s)}(z)\mathbf S^{(s-1)}(z^2),
\qquad s=1,\ldots,K.
\label{eq:PU_lifting_recursion}
\end{equation}
Each stage interleaves signed copies of the preceding row
coefficients, preserving the biphase alphabet and doubling
the sequence length.

Since each stage satisfies
$\mathbf T^{(s)}(z)\mathbf T^{(s)H}(z^{-1})=2\mathbf I_D$,
the resulting library $\mathbf S(z)=\mathbf S^{(K)}(z)$
has length $L=2^K$ and satisfies
\begin{equation}
\mathbf S(z)\mathbf S^H(z^{-1})
=
D2^K\mathbf I_D
=
DL\mathbf I_D,
\label{eq:PU_final_normalization}
\end{equation}
as required by \eqref{eq:PU_time}. The Sylvester Hadamard
initialization restricts $D$ to a power of two, with $D\ge2$.

\section{BTQ Realization and Complementary Correlation}
\label{app:MSKrealization}

For the biphase library used in the simulations, the BTQ realization
can be expressed using a common half-sine pulse,
\begin{equation}
h(t)=
\begin{cases}
\sin\!\left(\dfrac{\pi t}{T_c}\right), & 0\le t<T_c,\\
0, & \text{otherwise},
\end{cases}
\label{eq:btq_pulse}
\end{equation}
where $T_c$ denotes the paired-chip interval. The BTQ realization
in \eqref{eq:btq_waveform} places the even-indexed chips on the
sine rail and the odd-indexed chips on the cosine rail, offset
by $T_c/2$ \cite{TaylorBlinchikoff88}. On their respective
pulse supports,
\[
\begin{aligned}
\sin\!\left(\frac{\pi t}{T_c}\right)
&=(-1)^r h(t-rT_c),\\
\cos\!\left(\frac{\pi t}{T_c}\right)
&=(-1)^{r+1}h\!\left(t-\left(r+\frac12\right)T_c\right).
\end{aligned}
\]
Including the factor $-j$ on the sine rail, the pulse
coefficients are therefore $(-j)^{2r+1}\mathbf{s}_d[2r]$
and $(-j)^{2r+2}\mathbf{s}_d[2r+1]$, respectively.
Combining the even and odd indices gives
\begin{equation}
\mathbf{x}_d^{\mathrm{BTQ}}(t)
=
-j\sum_{\ell=0}^{L-1}
(-j)^\ell\mathbf{s}_d[\ell]
h\!\left(t-\frac{\ell T_c}{2}\right).
\label{eq:btq_pulse_expansion}
\end{equation}

Define the pulse autocorrelation by
\[
R_h(\tau)\triangleq
\int h(t)h^*(t-\tau)\,dt.
\]
Expanding the correlation of \eqref{eq:btq_pulse_expansion},
grouping terms by $k=\ell-m$, and using \eqref{eq:Cd_def} yields
\begin{equation}
\begin{aligned}
&\int
\mathbf{x}_d^{\mathrm{BTQ}}(t)
\mathbf{x}_d^{\mathrm{BTQ},H}(t-\tau)\,dt\\
&\quad=
\sum_{\ell,m}
(-j)^\ell j^m
\mathbf{s}_d[\ell]\mathbf{s}_d^H[m]\,
R_h\!\left(\tau-\frac{(\ell-m)T_c}{2}\right)\\
&\quad=
\sum_k(-j)^k\mathbf{C}_d[k]
R_h\!\left(\tau-\frac{kT_c}{2}\right).
\end{aligned}
\label{eq:btq_column_correlation}
\end{equation}
Thus, the discrete correlation matrices appear as phase-weighted
coefficients of shifted pulse autocorrelations.

Summing \eqref{eq:btq_column_correlation} over the library columns
and applying \eqref{eq:comp_set} gives
\begin{equation}
\begin{aligned}
&\sum_{d=0}^{D-1}\int
\mathbf{x}_d^{\mathrm{BTQ}}(t)
\mathbf{x}_d^{\mathrm{BTQ},H}(t-\tau)\,dt\\
&\quad=
\sum_k(-j)^k
\left(\sum_{d=0}^{D-1}\mathbf{C}_d[k]\right)
R_h\!\left(\tau-\frac{kT_c}{2}\right)\\
&\quad=
DL\,R_h(\tau)\mathbf{I}_D.
\end{aligned}
\label{eq:btq_complementarity}
\end{equation}
Hence, complementary combining cancels all off-diagonal channel
responses and leaves only the common pulse autocorrelation on
the diagonal. Since $R_h(\tau)=0$ for $|\tau|\ge T_c$, the
summed response has no delay sidelobes outside this mainlobe.

\bibliographystyle{IEEEtran}
\bibliography{DoppRef.bib}
\end{document}